\documentclass[12pt]{amsart}
\usepackage{amsmath}
\usepackage{amssymb}
\usepackage{amsthm}
\usepackage[onehalfspacing]{setspace}

\usepackage{natbib,verbatim,subfigure}
\usepackage{tikz}
\usetikzlibrary{snakes,calc,intersections}

\def\ep{\varepsilon}

\def\os{\varnothing}
\def\argmax{\text{argmax}}

\def\la{\lambda}
\def\da{\delta}
\def\Re{\mathbf{R}}

\newcommand{\df}[1]{\textit{#1}}

\newcommand{\ip}[1]{\langle #1 \rangle}
\newcommand{\abs}[1]{ \left | #1 \right | }

\newdimen\slantmathcorr
\def\oversl#1{%assuming that mathslant=0.25
\setbox0=\hbox{$#1$}
\slantmathcorr=\wd0
\hskip 0.2\slantmathcorr \overline{\hbox to 0.8\wd0{%
\vphantom{\hbox{$#1$}}}}
\hskip-\wd0\hbox{$#1$}
}

\def\undersl#1{%assuming that mathslant=0.25
\setbox0=\hbox{$#1$}
\slantmathcorr=\wd0
\underline{\hbox to 0.8\wd0{%
\vphantom{\hbox{$#1$}}}}
\hskip-0.8\wd0\hbox{$#1$}
}

\theoremstyle{plain}
\newtheorem{theorem}{Theorem}%[section]

\newtheorem{lemma}[theorem]{Lemma}

\newtheorem{corollary}[theorem]{Corollary}

\theoremstyle{definition}

\theoremstyle{remark}
\newtheorem{remark}[theorem]{Remark}

\newtheorem*{claim*}{Claim}

\newtheoremstyle{named}{}{}{\itshape}{}{\bfseries}{.}{.5em}{#1\thmnote{
    #3}}
\theoremstyle{named}

\newtheoremstyle{named2}{}{}{\itshape}{}{\bfseries}{:}{.5em}{#1\thmnote{
    #3}}
\theoremstyle{named2}

\begin{document}

\title{A characterization of combinatorial demand}

%\author[Chambers and Echenique]{Christopher P. Chambers \and Federico
 % Echenique}

\author[Chambers]{Christopher P. Chambers}

%\address[Chambers]{Department of Economics, University of California
%  at San Diego}

\author[Echenique]{Federico Echenique}

%\address[Echenique]{Division of the Humanities and Social Sciences,
%  California Institute of Technology}
%\email[A1,A2]{emails of Authors 1 and 2}

\thanks{Chambers is at the Department of Economics, University of California
  at San Diego. Echenique is at the Division of the Humanities and
  Social Sciences, California Institute of Technology. We wish to
  thank Itai Sher, Jay Sethuraman, and Renato Paes Leme for
  discussions on an earlier draft.} 

\date{\today}

\begin{abstract}
We prove that  combinatorial demand functions are characterized by two
properties: continuity and the law of demand.
\end{abstract}

\maketitle

\section{Introduction}

We prove that  combinatorial demand functions are characterized by two
properties: continuity and the law of demand. Suppose given a finite
collection of items. We are interested in the demand for
packages, or bundles  of items. For each vector of item prices,
we are given a collection of demanded packages, and we want to know if
there exists a valuation function for packages such that the demanded
packages are optimal. Utility is quasilinear in money.
So the valuation has to be such that, for each
price vector, the demanded packages maximize the value of the
packages when one subtracts the sum of the prices for the items in
the package.

The two properties that characterize optimal combinatorial demand are
upper hemicontinuity and the law of demand. The continuity property
is technical, but familiar. The law of demand captures the economic
nature of our problem. Demand for a single item must ``slope down,''
meaning that higher prices correspond to smaller demands. For
combinatorial demand, the law of demand says that the change in demanded
items should have a negative value, when evaluated by the change in
prices. The law of demand  is an aggregate, or average, version of the downward
sloping demand property, and it has a long history in economics (see, for
example, \citet{samuelson1948foundations}). 

In addition to our characterization of combinatorial demand functions,
we show that utility functions are uniquely identified by
combinatorial demand. Specifically, we show that, up to an additive constant, a
unique monotone utility function can be backed out from demand behavior.

While very natural, our result appears to be new. A long literature
investigates the combinatorial demands that satisfy specific
behavioral properties, such as  gross substitutes:
\citet{murota2003discrete}, \citet{tamura2004applications},
and \citet{leme2014gross} survey the
literature. Our result is more basic, in that we seek to understand
optimal demand behavior alone, without additional behavioral properties.
\citet{brown2007nonparametric}
investigate a similar question to ours in the context of bundles of
infinitely divisible goods, but their result does not extend to
combinatorial demand.  \citet{sher2014identifying} show that
aggregate combinatorial demand genericall identifies individual
valuation functions. 
Finally, we should mention
the paper by \citet{baldwin2012tropical} which introduces a new
framework for the study of discrete demand, investigates many of
its properties, and their implications for markets for discrete goods. 

Our main result (Theorem~\ref{thm:main}) follows along the lines of
Rochet's approach to revealed preference theory (see
\citet{rochet}). The property of cyclic monotonicity is crucial to
obtain a rationalizing valuation. We 
use the results of \citet{lavi2003towards} or \citet{saks} (in a
version due to \citet{ashlagi}) to
establish that the law of demand is sufficient for cyclic
monotonicity. The main issue in adapting these various results to our
problem is that cyclic monotonicity is not enough to obtain a strict
rationalization: the difficulty is that one may add optimal packages
when  constructing the rationalization from cyclic monotonicity. The
crucial idea to overcome this difficulty is contained in Lemma~\ref{lem:spade}
in the proof.

Our result on identification (Theorem~\ref{thm:identify}) is
essentially an adaptation of Theorem 24.9 in
\citet{rockafellar1970convex}. 

\section{Results}

\subsection{Notation:}
Let $X$ be a finite set. Let $S$ be the set of all nonempty subsets
of $2^X$ (so the empty set is not in $S$, but $\{\os\}$ is).

We identify a set $A\subseteq X$ with its indicator function $1_A\in
\Re^X$. The inner product of a vector $p\in \Re^X$ and $1_A$ is
denoted by $\ip{p,A} = \sum_{x\in A} p_x$.

\subsection{Rationalizable demand}

A \df{demand function} is a function
$D:\Re^X_{++}\rightarrow S$  with the property that there is $\bar
p\in \Re^X_{++}$ such that $D(p) = \{ \os\}$ for all $p\geq \bar p$.

The relevant properties for a demand function are three:
A demand function $D$
\begin{itemize}
\item is \df{quasilinear rationalizable} if there exists
$v:2^X\rightarrow \Re$ such that \[
D(p) = \argmax\{ v(A) - \ip{p,A} :A\subseteq X \};
\]
\item satisfies the  \df{law of demand} if for all
$p,q\in \Re^X_{++}$, and all $A\in D(p)$ and $B\in D(q)$, \[
\ip{p - q, A-B}\leq 0;
\]
\item is \df{upper hemicontinuous} if,
for all $p\in\Re^X_{++}$, there is a neighborhood $V$ of $p$ such that
$D(q)\subseteq D(p)$ when $q\in V$.
\end{itemize}

  \begin{theorem}\label{thm:main}
    A demand function is quasilinear rationalizable iff it is upper
    hemicontinuous and  satisfies  the law of demand.
  \end{theorem}

A stronger condition places more restrictions on the rationalization.  We say a function $g:\Re_+^X\rightarrow\Re$ is \emph{monotone} if for all $x,y\in\Re_+^X$, $x\leq y$ (coordinatewise) implies $g(x)\leq g(y)$.
$D$ is \df{monotone, concave, quasilinear rationalizable (MCQ-rationalizable)} if there exists a monotone, concave $g:\Re_{+}^X\rightarrow\Re$ such that $v(A)=g(1_A)$, and \[
D(p) = \argmax\{ v(A) - \ip{p,A} :A\subseteq X \}.
\]

An easy corollary, demonstrated by our proof is the following:

\begin{corollary}If a demand function is quasilinear rationalizable, then it is MCQ-rationalizable.\end{corollary}

The corollary demonstrates that there is no additional empirical content delivered by the hypotheses of concavity and monotonicity.

\subsection{Identification}

Say that $v:2^X\rightarrow \Re$ \emph{(quasilinear) rationalizes} $D$
if for all $p\in\Re_{++}^X$, 
$D(p)=\arg\max_{A\in S} v(A)-\ip{p,A}$.  
%Say that $v$ \emph{weakly
%  rationalizes} $f:\Re_{++}^X\rightarrow 2^X$ if for all
%$p\in\Re_{++}^X$, $f(p)\in \arg\max_{A\in S} v(A)-\ip{p,A}$.  

If $v$ rationalizes $D$, then so does $v+c$ for any constant function
$c$. So one can only hope to obtain identification up to an additive
constant, and that is indeed what one obtains.

We say that $v$ is \emph{monotone} if for all $A,B\in S$, if
$A\subseteq B$, then $v(A)\leq v(B)$. It is easy to see (see
Remark~\ref{rmk:monotone} below) that if $D$ is quasilinear rationalizable then
there is a monotone $v$ that quasilinear rationalizes it.

\begin{theorem}\label{thm:identify}
  For any quasilinear rationalizable $D$, there is a unique monotone
  $v$ for which $v(\varnothing)=0$ which rationalizes $D$. 
\end{theorem}

\section{Proof of Theorem~\ref{thm:main}}

\begin{lemma}\label{lem:necessity}
  If $D$ is quasilinear rationalizable then it is upper hemicontinuous
  and satisfies the law of demand
\end{lemma}

\begin{proof}
Let $v$ rationalize $B$. Let $u(p) = \max\{v(A) -\ip{p,A}:A\subseteq X\}$.

First we show that $D$ is upper hemicontinuous.
Since $X$ is finite, there is $\ep>0$ such
that $u(p)-(v(B')-\ip{p,B'}) > \ep$ for all $B'\notin D(p)$. Let $V$
be a ball with center $p$ and radius small enough that for all $q\in
V$, and all $B'\notin D(p)$,  $u(q)-(v(B')-\ip{q,B'}) > \ep$. Then
$D(q)\subseteq D(p)$ for all $q\in V$.

Second we show the law of demand. Let $A\in D(p)$ and $A'\in
D(p')$. Then
$v(A) -\ip{p,A}\geq v(A') -\ip{p,A'}$ and
$v(A') -\ip{p',A'}\geq v(A) -\ip{p',A}$. Adding these two inequalities
and rearranging yields $\ip{p-p',A-A'}\leq 0$.
\end{proof}

Lemma~\ref{lem:necessity} establishes the necessity direction in the
theorem. We now turn to sufficiency. The upper hemicontinuity of $D$
implies the following property:
A demand function $D$ satisfies condition $\spadesuit$ if for all $p$ and
$B\notin D(p)$ there is $A\in D(p)$ and  $p'$ such that  $A\in D(p')$ and
  $\ip{p',A-B} >\ip{p,A-B}$.

\begin{lemma}\label{lem:spade}
  If $D$ is upper hemicontinuous, then it satisfies condition $\spadesuit$.
\end{lemma}
\begin{proof} Let $p\in\Re^X_{++}$ and $B\notin D(p)$. Let $V$ be a
  neighborhood of $p$ as in the definition of upper hemicontinuity.  So
  $D(q)\subseteq D(p)$ for all $q\in V$.

Let $W=\cup_{A'\in D(p)} (A'\setminus B)$ and $E=\cup_{A'\in D(p)}
(B\setminus A')$. Note that
\begin{equation}
  \label{eq:3}
 (B\setminus A') \cup (A'\setminus B) \neq \os
\end{equation}
for each $A'\in D(p)$.

Let $\la,\la'>0$. By definition of $W$, $\ip{1_W,B}=0$. So
\[\ip{\la  1_W- \la' 1_E ,A'-B} = \la \ip{  1_W,A'}  - \la' \ip{ 1_E  ,A'}
+ \la' \ip{ 1_E  ,B}.
\]
Then for each  $A'\in D(p)$, \eqref{eq:3} implies that $\ip{  1_W,A'}\neq 0$ or $\ip{ 1_E
  ,B}\neq 0$, or both. Moreover, if $\ip{  1_W,A'} = 0$ then it must
be true that  $A'\subsetneq B$, which implies that
\begin{equation}
  \label{eq:1}
 -\ip{1_E, A' } + \ip{1_E, B } = \ip{1_E, B-A' } > 0.
\end{equation}

Choose $\la,\la'>0$ such that $\la \ip{  1_W,A'}  - \la' \ip{ 1_E  ,A'}
+ \la' \ip{ 1_E  ,B} >0$ for all $A'\in D(p)$. This is possible by
equation~\eqref{eq:1}, and for example by letting
$\la/\la'>\abs{X}$. Also choose $\la,\la'$  such that
 $p'=p+ (\la  1_W- \la' 1_E)\in V$.

Now, for any $A'\in D(p')$, \[\ip{p',A'-B}-\ip{p,A'-B} = \ip{(\la  1_W-
  \la' 1_E),A'-B} >0.\] Moreover, $A'\in  D(p)$, as $p'\in V$ and thus
$D(p')\subseteq D(p)$.
\end{proof}

A demand function satisfies \df{cyclic monotonicity} if, for all $n$,
and using summation mod $n$,
\[
\sum_{i=1}^n \ip{p_i , A_i - A_{i+1} }\leq 0, \] where $A_i\in
D(p_i)$, for all sequences $\{p_i\}_{i=1}^n$.

The following argument is mostly standard, adapting the construction of \citet{rockafellar} and \citet{rochet}.  A potential novelty is the use of the upper hemicontinuity condition in guaranteeing strict inequalities when necessary.

\begin{lemma}
  If $D$ satisfies cyclic monotonicity, and condition $\spadesuit$,
  then it is quasilinear rationalizable.
\end{lemma}

\begin{proof}We have assumed that there is $p^*$ for which $\{\varnothing\}=D(p^*)$.  For any $A\subseteq X$, define:

\[v(A)=\inf \ip{p_1,A-A_1} + \ldots + \ip{p^*,A_k-\varnothing},\]

where the infimum is taken over all finite sequences $(p_i,A_i)_{i=1}^k$ for which $A_i \in D(p_i)$.

Observe that by cyclic monotonicity, $v(\varnothing)\in \Re$; in fact $v(\varnothing)\geq 0$.  By construction, $v$ is nondecreasing, as it is the lower envelope of nondecreasing functions.  Hence $v(A)\in\Re$ for all $A$.  Finally, observe that $v$ is the lower envelope of restriction of affine functions on $\Re^X$.  Conclude that $v$ is the restriction of a concave function on $\Re^X$.

Finally, observe by construction that if $A\in D(p)$, then for any $B\subseteq X$,

\[v(B) \leq \ip{p,B-A} + v(A),\] from which we obtain $v(A)-\ip{p,A}\geq v(B) - \ip{p,B}$.

\begin{comment}
Let $R=\{A\subseteq X: \exists p \text{ s.t. } A\in D(p)\}$. Let $A\in
R$.
By cyclic monotonicity, the set $\{
\sum_{i=1}^n \ip{p_i , (A_i - A_{i+1}) } :  A_{n+1}=\os,  A_1 =A,
\text{ and } A_i\in D(p_i)
\}$ is bounded above. Then $v(A)$, defined by
 \[
v(A) = \sup \{
\sum_{i=1}^n \ip{p_i , A_i - A_{i+1} } :  A_{n+1}=\os, A_1 =A, \text{
  and } A_i\in
D(p_i) \}\] is well defined.

Let $A, B\in R$. Then for any sequence $\{(A_i,p_i)\}$ as
in the definition of $v(B)$, we have that
\begin{equation}
  \label{eq:2}
v(A) \geq  \ip{p,A-B} + \sum_{i=1}^n \ip{p_i , A_i - A_{i+1} },
\end{equation}
where $p$ is such that $A\in D(p)$. Since~\eqref{eq:2} holds for all
sequences in the definition of $v(B)$, we have that
$v(A) \geq \ip{p, A - B} + v(B)$, and thus that
$v(A)-\ip{p,A} \geq v(B)-\ip{p,B} $.

For $A\notin R$, let \[
v(A)=\max\{ v(A') : A'\in R \text{ and } A'\subseteq A \}.
\] Note that $\os\in R$ by hypothesis, and therefore $v(A)$ is well
defined.

Let $A\in R$ and let $p$ be such that $A\in D(p)$. Let $B\notin
R$. Then $v(B)=v(B')$ for some $B'\in R$ and $B'\subseteq B$. So
$v(A)-\ip{p,A}\geq v(B')-\ip{p,B'}\geq v(B)-\ip{p,B}$.

In particular, we have established that if $A,A'\in D(p)$ then
$v(A)-\ip{p,A} = v(A')-\ip{p,A'}$, and that if $A\in D(p)$ and
$B\notin D(p)$ then  $v(A)-\ip{p,A} \geq v(B)-\ip{p,B}$.
\end{comment}

Finally, to prove the lemma we need to show that if in addition $B\notin D(p)$
then $v(A)-\ip{p,A} > v(B)-\ip{p,B}$, or that $v(A) > \ip{p, A -
  B} + v(B)$.   By condition $\spadesuit$, there is
$A'\in D(p)$ and  $p'$ such that  $A'\in D(p')$ and
  $\ip{p',A'-B} >\ip{p,A'-B}$.

Suppose that $\{(A_i,p_i)\}$ is a sequence as
in the definition of $v(A')$.
Then
\[
v(B) \leq \ip{p', B-A' } + \sum_{i=1}^n \ip{p_i , A_i - A_{i+1} }
< \ip{p, B-A' }  + \sum_{i=1}^n \ip{p_i , A_i - A_{i+1} },\]
so $v(B)< \ip{p, B-A' } + v(A')$; and thus \[
v(A)-  \ip{p,A} = v(A')-  \ip{p,A'}  > v(B)-  \ip{p',B} .
\]
\end{proof}

We finish the proof by using a recent result in the mechanism design
literature, establishing conditions under which monotonicity (a
condition that coincides with the law of demand) implies cyclic
monotonicity: see \citet{lavi2003towards} and \citet{saks}.

\begin{lemma}
  A demand function satisfies cyclic monotonicity if it satisfies the
  law of demand.
\end{lemma}
\begin{proof}

So let $D$ satisfy the law of demand and suppose towards a
contradiction that there is a sequence $(p_i,A_i)_{i=1}^n$, with
$A_i\in D(p_i)$ and $\sum_{i=1}^n\ip{p_i, A_i - A_{i+1}} > 0$
(summation mod $n$), but no such sequence with $n\leq 2$.  Choose
such a sequence with minimal $n$, and observe that $n\geq 3$.

  For any selection $f(p)\in D(p)$, if $f$ is monotone then it is
  cyclically monotone, see \emph{e.g.} \citet{saks} or
  \citet{ashlagi}, Theorem S.7 in the supplementary
  material.\footnote{Technically, the \citet{ashlagi} result requires
    the output of $f$ to be a probability measure.  To modify the
    construction to fit our environment, simply let $y^*\not\in X$,
    and consider the set $Y\subseteq \Re^{X\cup\{y^*\}}$ given by
    $Y=\{(p,0):p\in \Re^X_{++}\}$.  Define the function
    $f^*:Y\rightarrow \Delta(X\cup\{y^*\})$ by
    $f^*(p,0)(x)=\frac{1_{x\in f(p)}}{|X|}$ and
    $f^*(p,0)(y^*)=1-\frac{|f(p)|}{|X|}$.  Observe that
    $\ip{(q,0),f^*(p,0)} = \ip{q,f(p)}\frac{1}{|x|}$, and therefore
    monotonicity
    of $f$ is equivalent to that of $f^*$ and cyclic monotonicity of
    $f$ is equivalent to that of $f^*$.} Since $D$ satisfies the law
  of
  demand, any selection $f$ is monotone, and therefore cyclically
  monotone.

If $p_i\neq p_j$ for all $i,j=1,\ldots,n$ with $i\neq j$, then we can choose a
selection $f$ of $D$ with $f(p_i) = A_i$, violating cyclic monotonicity of $f$, and hence contradicting the fact that it is monotone.

We now claim that in fact it is the case that $p_i \neq p_j$ for all $i\neq j$.

Observe first that if $p_i = p_{i+1}$ for some $i$, then
$\ip{p_i A_i- A_{i+1}} + \ip{p_{i+1} A_{i+1} - A_{i+2}}=
\ip{p_i A_i - A_{i+2}}$, implying the existence of a shorter
sequence, a contradiction.
%So let   $p_i \neq p_{i+1}$.

Suppose then that $p_i = p_j$. By the preceding, we know that $j=i+1$ is false, and $i=j+1$ is false. Without loss, suppose that $i=1$.  Then $j\neq n$ and $j\neq 2$.  Further,
$\ip{p_j , A_j- A_{j+1}} =  \ip{p_j, A_j- A_{1}} + \ip{p_1, A_1-
  A_{j+1}}$, so
\[
\begin{split}
0< \sum_{i=1}^n\ip{p_i, A_i - A_{i+1}}  =
\ip{p_1, A_1- A_{2}} + \cdots +
\ip{p_j, A_j- A_{1}} + \ip{p_1, A_1- A_{j+1}} \\
+ \ip{p_{j+1}, A_{j+1}- A_{j+2}} + \cdots + \ip{p_n , A_1 - A_n}.
\end{split}
\]
Consequently, either $\ip{p_1, A_1- A_{2}} + \cdots +
\ip{p_j, A_j- A_{1}} > 0$ or $\ip{p_1, A_1- A_{j+1}} + \cdots
+ \ip{p_{j+1}, A_{j+1}- A_{j+2}} + \cdots + \ip{p_n , A_1 - A_n}>0$.
In either case, we have demonstrated the existence of a shorter cycle
violating cyclic monotonicity, a contradiction.

%contradicting that there is no shorter sequence exhibiting a
%violation of cyclic monotonicity.
\end{proof}

\section{Proof of Theorem~\ref{thm:identify}}

As usual, $D(\Re^X_{++})=\bigcup_{p\in\Re^X_{++}}D(p)$ is the \df{range}
of $D$, and similarly for a function $f:\Re_{++}^X\rightarrow 2^X$.

First, we characterize those bundles which are demanded at some set of prices.

\begin{lemma}\label{lem:range}Given $v$ which rationalizes $D$, $A\in D(\Re^X_{++})$ iff for all $B\subset A$, $B\neq A$, $v(B)<v(A)$.\end{lemma}

\begin{proof}First suppose that $A\in D(p)$ for some $p$.  Then for all $B\subset A$, $B\neq A$, $v(A)-\ip{p,A}\geq v(B)-\ip{p,B}$, implying $v(A)\geq v(B)+\ip{p,A-B}>v(B)$.

Conversely, suppose that for all $B\subseteq A$, $B\neq A$, we have $v(B)<v(A)$.  We want to show that there is $p^*$ for which $v(A)-\ip{p^*,A}\geq v(B)-\ip{p^*,B}$.  To ensure this inequality is satisfied, we simply choose $p^*(x)$ small for $x\in A$ and large for $x\notin A$. \end{proof}

\begin{remark}\label{rmk:monotone}
By Lemma~\ref{lem:range}, it is straightforward to see that for any $h:S\rightarrow \Re$, there is a monotone $\overline{h}$ which rationalizes the same demand as $h$, namely, $\overline{h}$ is the smallest monotone function pointwise dominating $h$: \[\overline{h}(A)=\sup_{B\subseteq A}h(B).\]
\end{remark}

\subsection{Proof of Theorem~\ref{thm:identify}}
We proceed to show that if $v$ and $w$ are monotone and both
rationalize $D$, they differ by a constant. 

For all $p\in\Re^X_{++}$, let $f(p)\in D(p)$.
First we show that any two functions that rationalize $D$ must differ
by a constant on the range of $f$.
So let $v$ rationalize $D$.
  Define $U_v(p)=\sup_{A\in S}v(A)-\ip{p,A}$ (the indirect utility function).
  Observe that by definition, $U_v(p)\geq v(A)-\ip{p,A}$ for all
  $(p,A)$, and $A\in D(p)$ iff $v(A)-\ip{p,A}=U(p)$.  Consequently,
  $v(A)\leq U_v(p)+\ip{p,A}$ for all $(p,A)$, with equality iff $A\in
  D(p)$.  It follows that $A\in D(p)$ iff for all $q$,
  $U_v(p)+\ip{p,A}\leq U_v(q)+\ip{q,A}$; \emph{i.e.} $U_v(q)\geq
  U_v(p)+\ip{p-q,A}$; or $U_v(q)\geq U_v(p)+\ip{q-p,-A}$.  In other
  words, $-A$ is a subgradient of $U_v$ at $p$ iff $A\in D(p)$. In
  particular, $f(p)$ is a subgradient of $U_v$.

Observe that $U_v$ is convex, real-valued, and continuous, and defined on
$\Re^X_{++}$, an open domain.  Since $U_v$ is convex, for any $x_1,x_2$, the function
$h(\lambda)=U_v(\lambda x_2 + (1-\lambda)x_1)$ is convex.  Since $f$ is a subgradient of $U_v$ at each $p$, we
obtain that  \begin{align*}
  h(\lambda + \da) - h(\la)  & =
U_v(\lambda x_2 + (1-\lambda)x_1+\da (x_2 - x_1) ) - U_v(\lambda x_2 +
                               (1-\lambda)x_1  ) \\
&\geq \ip{ f(\lambda x_2 + (1-\lambda)x_1) , \da (x_2 - x_1) } \\
& = \da \ip{ f(\lambda x_2 + (1-\lambda)x_1) , (x_2 - x_1) }.
\end{align*}
Hence, $(h(\lambda + \da) - h(\la)  )/\da \geq \ip{ f(\lambda x_2 +
  (1-\lambda)x_1) , (x_2 - x_1) }$ if $\da>0$ and $(h(\lambda + \da) -
h(\la)  )/\da \leq \ip{ f(\lambda x_2 +   (1-\lambda)x_1) , (x_2 -
  x_1) }$ if $\da<0$. Thus,
 $h'_-(p)\leq f(p)\cdot (x_2-x_1)\leq h'_+(p)$

Observe also that $\int_0^1 h'_-(x)dx=\int_0^1 h'_+(x) dx$, by
\citet{rockafellar1970convex} 
Corollary 24.2.1.   So in particular $U_v(x_2)=U_v(x_1)+\int_0^1
f(p)\cdot (x_2-x_1)$. Recall that $f$ is an arbitrary selection from
$D$ and does not depend on $v$.
Suppose that $w$ rationalizes $D$, and define $U_w$ analogously to
$U_v$. Then we obtain that
\[U_v(x_2) - U_v(x_1) = \int_0^1 f(p)\cdot (x_2-x_1) = U_w(x_2) -
U_w(x_1).\] Thus
$U_v = U_w + c$, for some constant $c$.

Now for all $A\in f(\Re^X_{++})$,
$v(A)=\inf_{p}U_v(p)+\ip{p,A} = \inf_{p}U_w(p)+\ip{p,A} +c = w(A)
+c$. So we have shown that $v$ and $w$ differ at most by a constant on
the range of $f$. Again, since $f$ was arbitrary, $v$ and $w$ differ
at most by a constant on the range of $D$.

Recalling that there is always $p^*\in\Re^X_{++}$ for which $\varnothing\in D(p^*)$, by Lemma~\ref{lem:range} and monotonicity of $v$, for any $A$ not in the range of $D$, we have \[v(A)=\sup_{\{B\subset A: B\in D(\Re^X_{++})\}}v(B).\]  A similar equality holds for $w$, hence $v$ and $w$ differ by a constant.

Therefore there is a unique monotone $v$ with $v(\varnothing)=0$ which
rationalizes $D$.   

\bibliographystyle{ecta}
\bibliography{discretedemand}

\begin{thebibliography}{13}
\newcommand{\enquote}[1]{``#1''}
\expandafter\ifx\csname natexlab\endcsname\relax\def\natexlab#1{#1}\fi

\bibitem[\protect\citeauthoryear{Ashlagi, Braverman, Hassidim, and
  Monderer}{Ashlagi et~al.}{2010}]{ashlagi}
\textsc{Ashlagi, I., M.~Braverman, A.~Hassidim, and D.~Monderer} (2010):
  \enquote{Monotonicity and implementability,} \emph{Econometrica}, 1749--1772.

\bibitem[\protect\citeauthoryear{Baldwin and Klemperer}{Baldwin and
  Klemperer}{2012}]{baldwin2012tropical}
\textsc{Baldwin, E. and P.~Klemperer} (2012): \enquote{Tropical geometry to
  analyse demand,} Tech. rep., Working paper, Oxford University.

\bibitem[\protect\citeauthoryear{Brown and Calsamiglia}{Brown and
  Calsamiglia}{2007}]{brown2007nonparametric}
\textsc{Brown, D.~J. and C.~Calsamiglia} (2007): \enquote{The nonparametric
  approach to applied welfare analysis,} \emph{Economic Theory}, 31, 183--188.

\bibitem[\protect\citeauthoryear{Lavi, Mu'alem, and Nisan}{Lavi
  et~al.}{2003}]{lavi2003towards}
\textsc{Lavi, R., A.~Mu'alem, and N.~Nisan} (2003): \enquote{Towards a
  characterization of truthful combinatorial auctions,} in \emph{Foundations of
  Computer Science, 2003. Proceedings. 44th Annual IEEE Symposium on}, IEEE,
  574--583.

\bibitem[\protect\citeauthoryear{Murota}{Murota}{2003}]{murota2003discrete}
\textsc{Murota, K.} (2003): \emph{Discrete convex analysis}, SIAM.

\bibitem[\protect\citeauthoryear{Paes~Leme}{Paes~Leme}{2014}]{leme2014gross}
\textsc{Paes~Leme, R.} (2014): \enquote{Gross substitutability: an algorithmic
  survey,} Preprint.

\bibitem[\protect\citeauthoryear{Rochet}{Rochet}{1987}]{rochet}
\textsc{Rochet, J.-C.} (1987): \enquote{A necessary and sufficient condition
  for rationalizability in a quasi-linear context,} \emph{Journal of
  Mathematical Economics}, 16, 191--200.

\bibitem[\protect\citeauthoryear{Rockafellar}{Rockafellar}{1966}]{rockafellar}
\textsc{Rockafellar, R.} (1966): \enquote{Characterization of the
  subdifferentials of convex functions,} \emph{Pacific Journal of Mathematics},
  17, 497--510.

\bibitem[\protect\citeauthoryear{Rockafellar}{Rockafellar}{1970}]{rockafellar1970convex}
\textsc{Rockafellar, R.~T.} (1970): \emph{Convex analysis}, Princeton
  university press.

\bibitem[\protect\citeauthoryear{Saks and Yu}{Saks and Yu}{2005}]{saks}
\textsc{Saks, M. and L.~Yu} (2005): \enquote{Weak monotonicity suffices for
  truthfulness on convex domains,} in \emph{Proceedings of the 6th ACM
  conference on Electronic commerce}, ACM, 286--293.

\bibitem[\protect\citeauthoryear{Samuelson}{Samuelson}{1948}]{samuelson1948foundations}
\textsc{Samuelson, P.~A.} (1948): \emph{Foundations of economic analysis},
  Harvard University Press.

\bibitem[\protect\citeauthoryear{Sher and il~Kim}{Sher and
  il~Kim}{2014}]{sher2014identifying}
\textsc{Sher, I. and K.~il~Kim} (2014): \enquote{Identifying combinatorial
  valuations from aggregate demand,} \emph{Journal of Economic Theory}, 153,
  428--458.

\bibitem[\protect\citeauthoryear{Tamura}{Tamura}{2004}]{tamura2004applications}
\textsc{Tamura, A.} (2004): \enquote{Applications of discrete convex analysis
  to mathematical economics,} \emph{Publications of the Research Institute for
  Mathematical Sciences}, 40, 1015--1037.

\end{thebibliography}
\end{document}